\documentclass[a4paper]{amsart}
\usepackage{graphicx,verbatim} \usepackage{pgf,pdfpages, color}
\usepackage{rotating}\usepackage{lscape}
\usepackage{epsfig}
\usepackage{cancel}
\usepackage{subfigure}
\usepackage{bbm}
\usepackage{dsfont}
\usepackage[labelformat=empty]{caption}
\usepackage{tikz}
\usepackage[pagewise]{lineno}

\numberwithin{equation}{section}

\newtheorem{proposition}{Proposition}
\theoremstyle{definition}

\theoremstyle{remark}
\newtheorem{remark}{Remark}

\newcommand{\beq}{\begin{equation}}
\newcommand{\eeq}{\end{equation}}
\newcommand{\beqn}{\begin{equation*}}
\newcommand{\eeqn}{\end{equation*}}
\newcommand{\bea}{\begin{eqnarray}}
\newcommand{\eea}{\end{eqnarray}}
\newcommand{\bean}{\begin{eqnarray*}}
\newcommand{\eean}{\end{eqnarray*}}

\newcommand{\be}{\begin{enumerate}}
\newcommand{\ee}{\end{enumerate}}
\newcommand{\bi}{\begin{itemize}}
\newcommand{\ei}{\end{itemize}}
\newcommand{\bd}{\begin{description}}
\newcommand{\ed}{\end{description}}

\usepackage{color}
\newcommand{\red}{\textcolor[rgb]{1.00,0.00,0.00}}

\begin{document}

\title{ An arbitrage driven price dynamics of Automated Market Makers in the presence of fees
}

\author[1]{Joseph Najnudel}
\address[Joseph Najnudel]{School of Mathematics, University of Bristol, BS8 1TW, Bristol, United Kingdom}
\email{joseph.najnudel@bristol.ac.uk}

\author[2]{Shen-Ning Tung}
\address[Shen-Ning Tung]{Department of Mathematics, National Tsing Hua University,  Hsinchu 300, Taiwan
}
\email{tung@math.nthu.edu.tw}

\author[3]{Kazutoshi Yamazaki}
\address[Kazutoshi Yamazaki]{School of Mathematics and Physics, the University of Queensland, Brisbane QLD 4072 Australia
}
\email{k.yamazaki@uq.edu.au}

\thanks{
}

\author[4]{Ju-Yi Yen}
\address[Ju-Yi Yen]{Department of Mathematics, University of Cincinnati, Cincinnati, OH 45221, USA}
\email{ju-yi.yen@uc.edu}

\email{} \curraddr{}


\date\today

\begin{abstract}
We present a model for price dynamics in the Automated Market Makers (AMM) setting. Within this framework, we propose a reference market price following a geometric Brownian motion. The AMM price is constrained by upper and lower bounds, determined by constant multiplications of the reference price. Through the utilization of local times and excursion-theoretic approaches, we derive several analytical results, including its time-changed representation and limiting behavior.

\end{abstract}

\maketitle
\baselineskip 20pt


\section{Introduction}

Automated Market Makers (AMMs)  \cite{2021arXiv210308842C, gobet:hal-04131680} are innovative algorithms that utilize blockchain technology to automate the process of pricing and order matching on decentralized exchanges. Their foundation on blockchain and the employment of smart contracts enable users to buy and sell crypto assets securely, peer-to-peer, without the dependency on intermediaries or custodians.

A critical distinction between AMMs and traditional centralized limit order book models is their mechanism for price determination. While an order book model derives prices through the intentions of individual buyers and sellers, an AMM determines prices based on the available liquidity in its pool. This pool, called the ``liquidity pool", consists of funds deposited by users, known as liquidity providers (LPs). These providers ``lock in" varying amounts of two or more tokens into a smart contract, making them available as liquidity for other users' trades. The barrier to entry for becoming an LP is low; one merely needs a self-custody wallet and a stock of compatible tokens.

A critical cost encountered by LPs in AMMs is adverse selection. This challenge primarily stems from arbitrageurs capitalizing on disparities between the lagging prices within AMMs and the real-time market prices often observed on centralized exchanges. Prominent studies, including \cite{2022arXiv220806046M, 2023arXiv230514604M}, delve into this concern by quantifying the losses experienced by LPs due to arbitrage. They employ a metric known as `loss-versus-rebalancing' (LVR) to measure the impact.

This paper presents a model of AMM pricing dynamics based on the following assumptions:
\begin{enumerate}
\item[a)] A reference market with infinite liquidity and no trading costs exists, where the reference market price $p=(p_t)_{t \geq 0}$ follows a geometric Brownian motion.
\item[b)]  The AMM applies a fixed trading fee tier of $(1-\gamma)\%$, proportional to the trading volume.  This fee can range from 1 basis point (bp) to 100 basis points (bps). Notably, Uniswap v3 offers options such as 1 bp, 5 bps, 30 bps, and 100 bps. Our focus is on the worst-case scenario, assuming that all trades within the pool are motivated by arbitrage.
\item[c)] Arbitrageurs actively monitor the market, initiating trades whenever they identify arbitrage opportunities.
\end{enumerate}
These assumptions integrate elements from two sources: \cite{2022arXiv220806046M}, which focuses on continuous arbitrage without fees, and \cite{2023arXiv230514604M}, which considers discrete arbitrage with fees. In our model, we address continuous arbitrage while incorporating the impact of fees.

From assumptions a)--c), we deduce the following relationships between the reference market price $p$ and  AMM price $\tilde{p}$: 
\begin{itemize}
\item Under Assumptions a) and b), an arbitrage opportunity exists when $p_t < \gamma \tilde{p}_t$ or $p_t > \frac{1}{\gamma} \tilde{p}_t$ (see \cite[\S 2.4]{10.1145/3419614.3423251}). 
\item According to Assumption b), the AMM price $\tilde{p}_t$ remains stable within the range of $\gamma p_t \leq \tilde{p}_t \leq \frac{1}{\gamma} p_t$, termed the no-arbitrage interval. This follows immediately from the first assertion.
\item Assumption c) implies that arbitrage actions occur only when $p_t = \gamma \tilde{p}_t$ or $p_t = \frac{1}{\gamma} \tilde{p}_t$, hence $\tilde{p}$ changes only at these times.
\end{itemize}

\begin{figure}[h]
\centering
\includegraphics[width=1.0\textwidth]{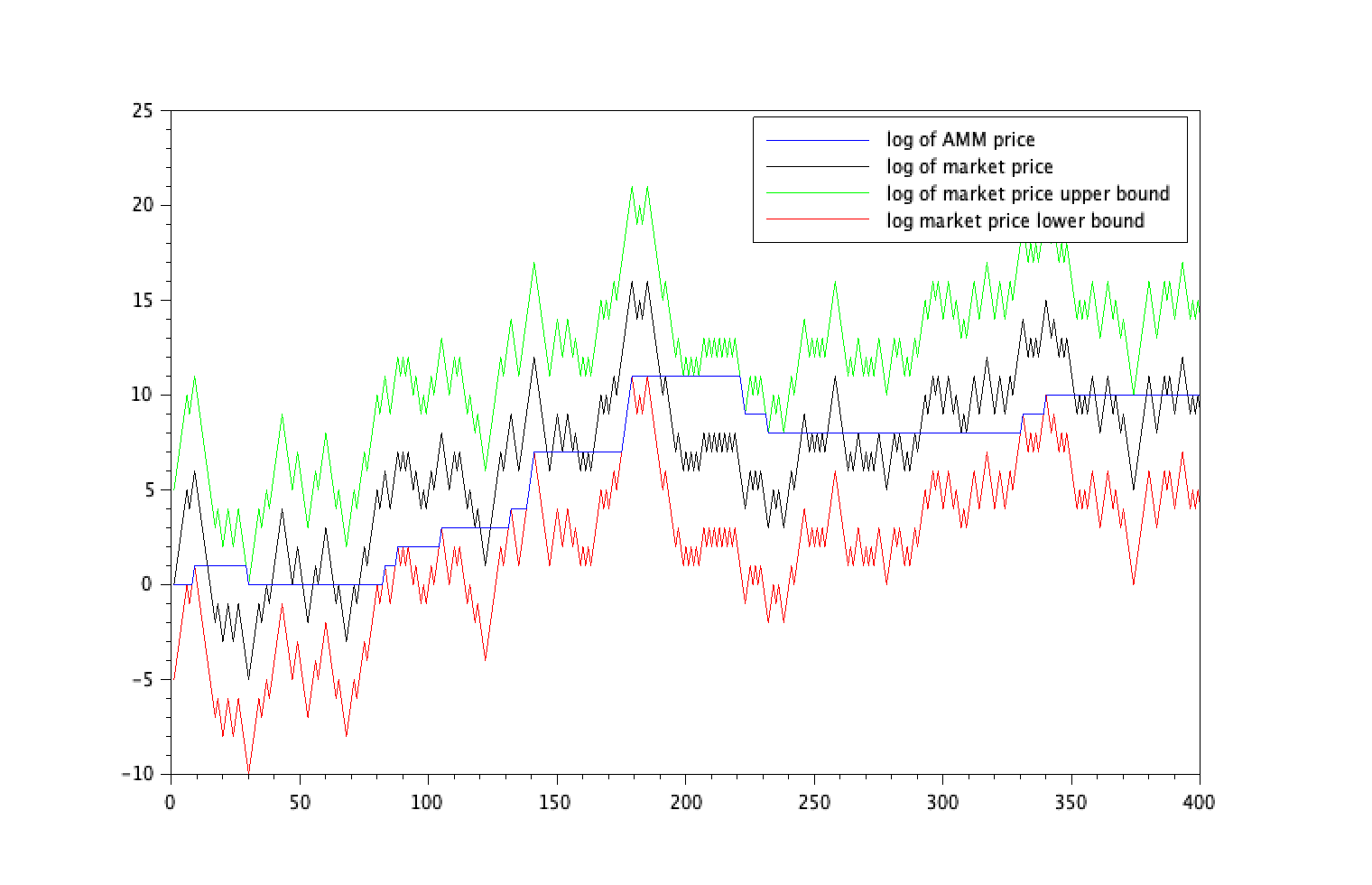}
\caption{Figure 1} 
\label{fig1}
\end{figure}

In Figure \ref{fig1}, a sample path of the logarithm of  $\tilde{p}$ is provided, along with the logarithm of  $p$ and the lower and upper bounds $\log (\gamma p)$ and $\log (\gamma^{-1} p)$, respectively. It is pushed upward (resp., downward) at times when it coincides with the lower (resp., upper) bound, and it stays constant at all other times.

Our work is inspired by the unpublished  paper by Tassy and White \cite{tassywhite}. While the context in the cited work specifically addresses Constant Product Market Makers, our model can be directly applied to all AMMs.
The key contribution of this paper is the detailed probabilistic description of the considered AMM price process, underpinned by the local time of the Brownian motion. The logarithm of $\tilde{p}$ can be seen as a concatenation of the running infimum and supremum processes of the market price, with appropriate shifting. By leveraging the classical results on the Skorokhod equation of a reflected Brownian motion, expressed in terms of Brownian local times, several analytical results concerning $\tilde{p}$ can be derived. Of particular interest is the well-posedness of such a process, its time-changed representation, and its asymptotic behavior. To the best of our knowledge, this paper represents the first study of this stochastic process, motivated by its applications in AMM.

The remainder of the paper is structured as follows: First, we present a precise mathematical construction of our AMM price process. Subsequently, we explore two distinct approaches to analyze the process under consideration in this study.

\section{AMM price process}

We assume the market price $p = (p_t)_{t \geq 0}$ (normalized in such a way that it is equal to $1$ at time $0$, and with the appropriate unit of time) is given by a geometric Brownian motion $p = \exp B$ for a standard Brownian motion  $B=(B_t)_{t \geq 0}$,  and the AMM price $\tilde{p}$ is constrained by the inequalities 
$$\gamma p  \leq \tilde{p} \leq \gamma^{-1}p$$
where the parameter $\gamma \in (0,1)$
 is related to the fees of the AMM. 
These inequalities are equivalent to 
\beq\label{ineq}
 B_t - c \leq  U_t   \leq  B_t + c 
 \eeq
where, for $t \geq 0$, $U_t$ is $\log \tilde{p}$ at time $t$, $B_t$ is $\log p$ at time $t$, and $c := \log (\gamma^{-1}) > 0$; see Figure \ref{fig1}.


The process $U=(U_t)_{t \geq 0}$ is chosen in such a way that $U_0 = B_0 = 0$,  and $U$ remains constant at any time where this is compatible with the inequalities \eqref{ineq}. More precisely, it is constant except when it is pushed up by the process $(B_t - c)_{t \geq 0}$ or pushed down by $(B_t +c)_{t \geq 0}$  in order to preserve the inequalities \eqref{ineq}. 

The precise description is described as follows, using a sequence of stopping times. If $B$ reaches $c$ before $-c$ (we call this event $A$), then we define, by induction: 
\begin{itemize}
\item $T_0$ as the infimum of $t \geq 0$ such that $B_t = c$. 
\item For all $k \geq 0$, $T_{2k+1}$ is the first time $t \geq T_{2k}$ such that
$B_t - \sup_{T_{2k} \leq s \leq t} B_s
= -2c $. 
\item For all $k \geq 0$, $T_{2k+2}$ is the first time $t \geq T_{2k+1}$ such that
$B_t - \inf_{T_{2k+1} \leq s \leq t} B_s
= 2c $. 
\end{itemize}
By the well-known property of the reflected Brownian motion, $T_{m+1}-T_m$ is finite, and so is $T_m$ almost surely (a.s.) for each $m \geq 0$.
Then, we define $(U_t)_{t \geq 0}$ as follows: 
\begin{itemize}
\item $U_t = 0$ for $0 \leq t \leq T_0$. 
\item For all $k \geq 0$ and $T_{2k} < t \leq T_{2k+1}$, $U_t = \sup_{T_{2k} \leq s \leq t} B_s - c$.
\item For all $k \geq 0$ and $T_{2k+1} < t \leq T_{2k+2}$, $U_t = \inf_{T_{2k+1} \leq s \leq t} B_s + c$.
\end{itemize}
Informally, $U$ is pushed up by $B - c$ in the intervals of time $[T_{2k}, T_{2k+1}]$ and pushed down by $B+c$ in the intervals of time $[T_{2k+1}, T_{2k+2}] $. 
Similarly, if 
$B$
reaches $-c$ before $c$ (event $A^c$), then we define: 
\begin{itemize}
\item $T_0$ as the infimum of $t \geq 0$ such that $B_t = -c$. 
\item For all $k \geq 0$, $T_{2k+1}$ is the first time $t \geq T_{2k}$ such that
$B_t - \inf_{T_{2k} \leq s \leq t} B_s
= 2c $. 
\item For all $k \geq 0$, $T_{2k+2}$ is the first time $t \geq T_{2k+1}$ such that
$B_t - \sup_{T_{2k+1} \leq s \leq t} B_s
= -2c $. 
\end{itemize}
Then,
\begin{itemize}
\item $U_t = 0$ for $0 \leq t \leq T_0$. 
\item For all $k \geq 0$ and $T_{2k} < t \leq T_{2k+1}$, $U_t = \inf_{T_{2k} \leq s \leq t} B_s +c$.
\item For all $k \geq 0$ and $T_{2k+1} < t \leq T_{2k+2}$, $U_t = \sup_{T_{2k+1} \leq s \leq t} B_s - c$.
\end{itemize}

\begin{proposition} \label{Skorokhod}
The process $U:= (U_t)_{t \geq 0}$ is the only continuous process with finite variation satisfying the following properties 
\begin{enumerate}
   \item  $U_0 = 0$,
\item $B_t - c \leq U_t \leq B_t + c$ for all $t \geq 0$,
    \item $U$ is nondecreasing in all intervals where 
it remains strictly smaller than $B+c$,
\item $U$ is nonincreasing
in all intervals where 
it remains strictly larger than $B-c$.
\end{enumerate}
 Notice that the last two properties imply that $U$ is constant on intervals where $B-c < U< B + c$.
\end{proposition}
\begin{proof}
The process $U$ is clearly continuous on $\mathbb{R} \backslash (T_m)_{m \geq 0}$.
A careful look at the definition above 
shows that $U$ has no discontinuity at times $(T_m)_{m \geq 0}$. 
Indeed, on $A$,
\begin{align} \label{U_cont}
\begin{split}
\lim_{t \uparrow T_{2k+1}}U_t &= \lim_{t \uparrow T_{2k+1}} \left(\sup_{T_{2k} \leq s \leq t} B_{s} - B_t \right) +  \lim_{t \uparrow T_{2k+1}}  B_t - c = 2c + B_{T_{2k+1}} - c =U_{T_{2k+1}}, \\
\lim_{t \uparrow T_{2k+2}}U_t &= \lim_{t \uparrow T_{2k+2}} \left(\inf_{T_{2k+1} \leq s \leq t} B_{s} - B_t \right) + \lim_{t \uparrow T_{2k+2}} B_t + c = - 2 c + B_{T_{2k+2}} + c =U_{T_{2k+2}},
\end{split}
\end{align}
where the second equalities hold because the supremum and infimum are attained at the hitting times $T_{2k+1}$ and $T_{2k+2}$ of reflected Brownian motion. The continuity holds in the same way on $A^c$ as well.

Since $U$ is monotone on each interval of the form $[T_m, T_{m+1}] $ for $m \geq 0$, it has finite variation. 

We now show properties (a)-(d) focusing on the case $A$.
The proof for the other case is similar.
We have \begin{itemize}
\item[(i)] $U$ is equal to zero on $[0, T_0]$. 
\item[(ii)] For all $k \geq 0$, $U$ is nondecreasing in $[T_{2k}, T_{2k+1}]$ and nonincreasing in $[T_{2k+1}, T_{2k+2}]$. 
\item[(iii)] For all $k \geq 0$, $U$ is equal to $B-c$ at time $T_{2k}$ and to $B+c$ at time $T_{2k+1}$.
\end{itemize}
Indeed, (i) holds by the definition of $U$ and (iii) is clear from the definitions of the stopping times $(T_m)_{m \geq 0}$ (see also \eqref{U_cont}). (ii) holds by the definition of $U$ in terms of a piecewise supremum and infimum process, which is piecewise monotone. Now (a) holds by (i) and (b) holds by (ii) and (iii).

It remains to show (c) and (d); we only provide proof for (d) as that for (c) is similar. Let $J := \{t \geq 0: U_t > B_t -c \}$.
By (iii), $J \subset [0, \infty) \backslash (T_{2k})_{k \geq 0}$, which is the union of $[0, T_0)$ and $(T_{2k}, T_{2k+2}), k \geq 0$. We know from (i) that $U$ is constant on $[0, T_0]$, and from (ii),  that $U$ is nonincreasing in $[T_{2k+1}, T_{2k+2})$. It is then enough to consider intervals 
\[
I \subset J \cap \left(\bigcup_{k \geq 0} (T_{2k}, T_{2k+1}] \right)
\]
a sub-interval of $(T_{2k}, T_{2k+1}]_{k \geq 0}$ at which $U_t > B_t-c$.   For $t \in I$, for some $k \geq 0$, $U_t = \sup_{T_{2k} \leq s \leq t} B_s - c$ and $U_t > B_t - c$, thus, the supremum of $B$ on $[T_{2k}, t]$ is not reached at $t$. This implies that $\sup_{T_{2k} \leq s \leq t} B_s$ remains constant for $t \in I$, and thus $U$  is constant, and then nonincreasing, in $I$, completing the proof of (d).

It remains to prove that the properties stated in Proposition \ref{Skorokhod} uniquely determine $U$. 

Let $U^1$ and $U^2$ be two distinct processes satisfying the conditions (a)-(d), and let 
\begin{align}
S:= \inf \{ t \geq 0: U^1_t \neq U^2_t\}. \label{S_def}
\end{align}
We assume $S < \infty$ and derive contradiction.

Necessarily, by (a), (c) and (d), $U^1 = U^2 = 0$ until the first time where $B$ reaches $\{-c,c\}$ which is strictly positive, and thus $S > 0$. By the definition of $S$ as in \eqref{S_def}, $U^1_t = U^2_t$ for all $t \in [0, S)$, and by the continuity of $U^1$ and $U^2$, we also have $U^1_S = U^2_S =: U_S$. 

Suppose 
\begin{align}
 U^1_S - B_S = U^2_S - B_S = U_S - B_S \geq 0. \label{assump_contradiction}  
\end{align}
By this, continuity of $B$ and $U$ and (b), we have
$B_t - c < U^j_t  \leq B_t + c$ for $j = 1,2$ on the interval $[S, S+u]$ for sufficiently small $u > 0$. 
For $j \in \{1,2\}$, let us define the following processes 
on $[S, S+u]$: 
\begin{align*}
A_t^j &:= U_S - U^j_t, \\
Y_t &:= B_t + c - U_S, \\
Z_t^j &:= Y_t + A_t^j = B_t + c - U^j_t.
\end{align*}
By (b), $Y_S \geq 0$ and $Z_t^j \geq 0$, and $A_t^j$ is continuous and vanishes at the starting time $S$.  

Moreover, the variation of $A^j$ is carried by the set of times where $Z^j = 0$. Indeed, $A^j$ changes only when $U^j$ changes, 
or equivalently, only when
$U^j$ is equal to $B-c$ or $B+c$,  and then only when $U^j = B+c$ since we know that $U^j > B-c$ on the interval $[S, S+u]$. From this last property and (d),
$U^j$ is nonincreasing in $[S, S+u]$, which implies that $A^j$ is nondecreasing in this interval.  
By Skorokhod's lemma (see Lemma 2.1, Chapter VI of \cite{RY}; Lemma 3.1.2, Chapter 3 of \cite{MR3134857}), $A^j$ is uniquely determined by $Y$ on the interval $[S, S+u]$, and thus $A^1 = A^2$, which implies that $U^1 = U^2 $ on the interval $[S, S+u]$. This contradicts the fact that $S$ is the infimum of times where $U^1 \neq U^2$ as in \eqref{S_def}. This contradiction shows that $S = \infty$, i.e.\
one cannot have two distinct processes satisfying the assumptions of Proposition \ref{Skorokhod}. Similar argument holds for the complement case of \eqref{assump_contradiction}.

\end{proof}

Now, let $W= (W_t)_{t \geq 0}$ be another standard Brownian motion, and let $F: \mathbb{R} \to [-c,c]$ be the triangle wave function with period $4c$ such that
\[
F(x) = \left\{ \begin{array}{ll} \cdots & \cdots \\ -2c - x & \textrm{for } x \in [-3c, -c] 
\\
x & \textrm{for }  x \in [-c, c] \\ 2c - x & \textrm{for }  x \in [c, 3c]\\ \cdots & \cdots
\end{array} \right.
\]

\noindent The derivative of $F$ is $1$ in the interval $((4k-1)c, (4k+1)c)$ and  $-1$ in the interval $((4k+1)c, 4k+3)c)$, 
for all $k \in \mathbb{Z}$. 
The second derivative of $F$ in the sense of the distribution is 
$$2 \sum_{k \in \mathbb{Z}} (\delta_{ (4k-1)c } -
\delta_{ (4k+1)c })
$$
where $\delta_x$ denotes the Dirac measure at $x$. 
By It\^o-Tanaka formula, 
$$F(W_t) = \beta_t - V_t, \quad t \geq 0,$$
where $(\beta_t)_{t \geq 0}$ is a standard Brownian motion, 
$$V_t := \sum_{k \in \mathbb{Z}} 
(L^{(4k+1)c}_t - L^{(4k-1)c}_t), \quad t \geq 0,
$$
with $L^x_t$ denoting the local time of $W$ at time $t$ and level $x \in \mathbb{R}$. 

\begin{remark} 
The processes $(\beta_t)_{t \geq 0}$ and $(V_t)_{t \geq 0}$ are continuous and enjoy the following properties: 
\begin{enumerate}
\item $V_0 = 0$. 
\item $(V_t)_{t \geq 0}$ has paths of  finite variation, since it is the difference of two nondecreasing processes, given by sums of local times at different levels. 
\item For all $t \geq 0$, $  V_t - \beta_t = - F(W_t) \in [-c, c]$, so $\beta_t - c \leq V_t \leq \beta_t + c$. 
\item In an interval  $K
:=\{t: V_t < \beta_t + c\} = [0,\infty) \backslash \{t: F(W_t)= -c\}
$,  
$W$ does not reach 
levels congruent to $-c$ modulo $4c$ (these levels corresponding to values of $x \in \mathbb{R}$ such that $F(x) = -c$). Local times at levels $(4k-1)c$ for $k \in \mathbb{Z}$ are constant on $K$, which implies, from the definition of $V$, that $V$ is nondecreasing on $K$. 
\item In an interval $K'
:= \{t: V_t > \beta_t - c\} = [0,\infty) \backslash \{t: F(W_t)=c\}
$,
$W$ does not reach 
levels congruent to $c$ modulo $4c$. Local times at these levels are constant on $K'$, which implies that $V$ is nonincreasing on $K'$. 
\end{enumerate}
\end{remark}
From Proposition \ref{Skorokhod}, and the corresponding uniqueness property, 
the process $(V_t)_{t \geq 0}$ can be recovered from $(\beta_t)_{t \geq 0}$ in the same way as $(U_t)_{t \geq 0}$ has been constructed from $(B_t)_{t \geq 0}$. Since $(B_t)_{t \geq 0}$ and $(\beta_t)_{t \geq 0}$ are both Brownian motions, $(B_t, U_t)_{t \geq 0}$ has the same joint distribution as $(\beta_t, V_t)_{t \geq 0}$. 

We provide two descriptions of the processes discussed here using different time changes, the inverse local times and the hitting times.
\subsection*{Time change with inverse local times}
We define the following linear combination of local times: 
$$\left(\mathcal{L}_t := \sum_{k \in \mathbb{Z}} 
L^{(2k+1)c}_t \right)_{t \geq 0}, $$
and its inverse process $(\sigma_{\ell})_{\ell \geq 0}$ given by 
$$\sigma_{\ell} := \inf\{ t \geq 0, \mathcal{L}_t > \ell\}. $$
Note that $\mathcal{L}_t \nearrow \infty$ as $t \to \infty$, thus $\sigma_{\ell} < \infty$ for all $\ell \geq 0$ a.s.

Since $\mathcal{L}$ is an additive functional of the Brownian motion $W$ (see \cite{RY}, Chapter X), by the strong Markov property, the time-changed Brownian motion $(W_{\sigma_{\ell}})_{\ell 
\geq 0} $ is a continuous-time Markov chain. Moreover, since the variation of $(\mathcal{L}_t)_{t \geq 0}$ is supported on the set of times where $W$ is in $E := \{ (2k+1) c, k \in \mathbb{Z}\}$, we have that the Markov chain $(W_{\sigma_{\ell}})_{\ell 
\geq 0} $ takes its value only in $E$. This process can only jump between consecutive values in $E$. Indeed, if 
$(W_{\sigma_{\ell}})_{\ell 
\geq 0} $ jumps between non-consecutive values $x$ and $y$ in $E$ at some $\ell \geq 0$, then the path 
$(W_t)_{\sigma_{\ell-} \leq t \leq \sigma_{\ell}} $
starts at $x$, ends at $y$, without accumulating local time at levels in $E$ strictly between $x$ and $y$: this event occurs with probability zero. Similarly, the starting point
$W_{\sigma_{0}} $ can only be $c$ or $-c$. Thus, the following results are immediate.

\begin{proposition} 
The process $(W_{\sigma_{\ell}})_{\ell 
\geq 0} $ satisfies the following properties: 
\begin{enumerate}
\item It starts at $c$ with probability $1/2$, and at $-c$ with probability $1/2$. 
\item It jumps by $2c$ or $-2c$, the rate of all possible jumps being the same, again by symmetry. 
\end{enumerate}
\end{proposition}

Now, let us compute the rate of jumps. By translation invariance and the strong Markov property, the probability of having no jump during an interval of length $\ell$ is equal, for a Brownian motion $B$, its local time $L^B$  at level zero and its inverse local time $\tau_\ell := \inf \{ t \geq 0: L^B > \ell \}$, $\ell \geq 0$ : 
\bean
&&\mathbb{P}[\max\{|B_s|, s\le \tau_\ell\}<2c]\\
&&=\left(\mathbb{P}[\max\{B_s, s\le \tau_\ell\}<2c]\right)^2 \ \text{since the negative and positive excursions of $B$ are independent}\\
&&= \left(\mathbb{P}[L^B_{\inf\{t \geq 0, B_t = 2c\}} >\ell]\right)^2=\left(\mathbb{P}[\text{BESQ}_2 (2c)>\ell]\right)^2 \ \text{by Ray-Knight theorem}
\eean
where $\text{BESQ}_2$ denotes a squared Bessel process of dimension $2$. 
Hence, for a standard exponential variable ${\bf e}$, 
\bean
&&\mathbb{P}[\max\{|B_s|, s\le \tau_\ell\}<2c] \\
&&=\left(\mathbb{P}[2\times {\bf e}\times 2c>\ell]\right)^2 \\
&&=\left(e^{-\ell/4c}\right)^2 =e^{-\ell/2c}.
\eean
Hence, the sum of the rates of jump by $2c$ and $-2c$
is equal to $1/2c$, and then each of the two possible jumps has a rate 
$1/4c$ by symmetry. 

We have 
$$V_{\sigma_{\ell}} = \sum_{k \in \mathbb{Z}}
(L_{\sigma_{\ell}}^{(4k+1)c} -  L_{\sigma_{\ell}}^{(4k-1)c})
$$
and 
$$ \mathcal{L}_{\sigma_{\ell}} = \sum_{k \in \mathbb{Z}}
(L_{\sigma_{\ell}}^{(4k+1)c} + L_{\sigma_{\ell}}^{(4k-1)c}).
$$
In an interval of values of $\ell$ where $W_{\sigma_{\ell}}$ remains constant and equal to $ (4k-1)c$ for a given $k \in \mathbb{Z}$, the variations of $V_{\sigma_{\ell}}$ are opposite to the variations of $\mathcal{L}_{\sigma_{\ell}} = \ell$, since the only local time which varies in the corresponding sums is the local time 
at level $(4k-1)c$. Similarly, in an interval of values of $\ell$ where $W_{\sigma_{\ell}}$ remains constant and equal to $ (4k+1)c$ for a given $k \in \mathbb{Z}$, the variations of $V_{\sigma_{\ell}}$ are equal to the variations of $\ell$. 
Hence, $(V_{\sigma_{\ell}})_{ \ell \geq 0}$ is the difference of time spent by the Markov process $(W_{\sigma_{\ell}})_{\ell \geq 0}$ at levels congruent to $c$ modulo $4c$, minus the time spent at levels congruent to $-c$ modulo $4c$. 
  If we reduce $(W_{\sigma_{\ell}})_{\ell \geq 0}$ modulo $4c$ and divide it by $c$, we get a Markov chain, $(Y_\ell)_{\ell \geq 0}$, on the two-state space 
$\{-1, 1\}$, with its starting point uniform in this space, 
and rate of jumps equal to $1/2c$. Then, $(V_{\sigma_{\ell}})_{ \ell \geq 0}$ is the difference of time spent at state $1$ by this Markov process, minus the time spent at state $-1$, i.e. the integral of the Markov chain: $V_{\sigma_{\ell}} = \int_0^\ell (1_{\{Y_s = 1\}} - 1_{\{Y_s = -1\}}) ds$.

\subsection*{Time change with hitting times}
Let us define the sequence of stopping times 
$(H_k)_{k \geq 0}$, as follows: 
\bi
\item If $W$ hits $c$ before $-c$ (we call this event $\mathcal{A}$), $H_0=H_1$ is the  first hitting time of $c$, and for $k\ge1$,
$$H_{k+1}=\inf\{t>H_k, |W_t-W_{H_k}|\ge 2c\}.$$

\item If $W$ hits $-c$ before $c$ (or $\mathcal{A}^c$), $H_0$ is the first hitting time of $-c$, and for $k\ge0$,
$$H_{k+1}=\inf\{t>H_k, |W_t-W_{H_k}|\ge 2c\}.$$

\ei

Between $H_0$ and $H_1$, $V$ varies as the opposite of the local time of $W$ at $-c$; for $k \geq 1$, between $H_{2k-1}$ and $H_{2k}$, $V$ varies as the local time of $W$ at $W_{H_{2k- 1}}$; for $k \geq 1$, 
between $H_{2k}$ and $H_{2k+1}$, it varies as the opposite of the local time of $W$ at $W_{H_{2k}}$.
We deduce, for all $m \geq 1$, 
$$V_{H_{m}} = \sum_{j=0}^{m-1} (-1)^{j-1} \left(L^{W_{H_j}}_{H_{j+1}} - L^{W_{H_j}}_{H_{j}} \right), 
$$
where we recall $L^x_t$ denoting the local time of $W$ at time $t$ and level $x$. 
Notice that the term $j = 0$ vanishes on $\mathcal{A}$, since $H_0 = H_1$ in this case.

All these increments of local time, except the first one if $H_0=H_1$ 
in the event $\mathcal{A}$, are independent and have the distribution of the local time $L^B_{T^*_{2c}}$ of $B$ at level zero, where $B$ is a Brownian motion and $$T^*_{2c}=\inf\{t>0, |B_t|\ge2c\}.$$

Now, we show that using this approach,
one can compute this distribution via Ray-Knight Theorem:
for $\ell \geq 0$, 
$L^B_{T^*_{2c}}>\ell$ if and only if no excursion of $B$ before the inverse local time $\tau_\ell$ reaches $2c$ or $-2c$, which implies 
\bean
&&\mathbb{P}(L^B_{T^*_{2c}}>\ell)=\mathbb{P}(\text{no excursion of $B$ reaches } \pm 2c \text{ before } \tau_\ell)\\
&&=\left(\mathbb{P}(\text{no excursion of $B$ reaches }  2c \text{ before } \tau_\ell)\right)^2 \ \text{by independence of excursions}\\
&& =\left(\mathbb{P}(\inf\{ t \geq 0, B_t = 2c\} >\tau_\ell)\right)^2 \\
&&=\left(\mathbb{P}(L^B_{\inf\{ t \geq 0, B_t = 2c\}}>\ell)\right)^2\\
&&=\left(\mathbb{P}(4c {\bf e}>\ell)\right)^2 \text{ by Ray-Knight Theorem }\\
&&=e^{-\ell/2c}.
\eean
Hence, $L^B_{T^*_{2c}}$ has the same distribution as 
$2c {\bf e}$. We deduce the equality in distribution, for $m \geq 2$,  
$$V_{H_m} = 2 c \left( -  X {\bf e}_0
+ \sum_{j=1}^{m-1} (-1)^{j-1}{\bf e}_j \right), $$
where $({\bf e}_j)_{j \geq 0}$ is a sequence of i.i.d. 
standard exponential variables, and $X$ is an independent, Bernoulli random variable with parameter $1/2$: $X = 0$ if and only if $W$ hits $c$ before $-c$ (i.e.\ $\mathcal{A}$). 

For $k \geq 1$, $V_{H_{2k+1}}$ is a sum of 
$- 2c X {\bf e}_0$ and of $k$ i.i.d. variables with 
the same distribution as $2c ({\bf e}_1 - {\bf e}_2)$.
These $k$ variables are centered, with variance equal 
to $8c^2$, since ${\bf e}_1$ and ${\bf e}_2$ are independent with variance $1$. 
We get the following central limit theorem: 
\begin{proposition} \label{prop_CLT}
When $k$ tends to infinity, 
$$\frac{V_{H_{2k+1}}}{c\sqrt{8k}} \longrightarrow \mathcal{N}(0,1)$$
in distribution. 
\end{proposition}
 
The increment $H_{2k+1} - H_1$ is itself the sum of $2k$ i.i.d. random variables with the same distribution as $T^*_{2c}$. 
By the law of large numbers, we have a.s.
$$\frac{ H_{2k+1}}{2k} 
\underset{k \rightarrow \infty}{\longrightarrow}
\mathbb{E}[T^*_{2c}]. 
$$
Stopping the martingale $(B^2_t - t)_{t \geq 0}$ at time 
$T^*_{2c}$, we deduce for all $t \geq 0$, by optional sampling, that
$$\mathbb{E} [ B^2_{\min(t,T^*_{2c})} ] 
= \mathbb{E} [ \min (t, T^*_{2c})]. $$
For $t \rightarrow \infty$, $B^2_{\min(t,T^*_{2c})} \leq 4c^2$
almost surely converges to $4c^2$, and 
$\min(t, T^*_{2c})$ converges to $T^*_{2c}$ from below. Applying dominated convergence to the left-hand side and monotone convergence to the right-hand side, we deduce 
$$ 4c^2 = \mathbb{E} [ T^*_{2c}],$$
and then almost surely 
\begin{align}
\frac{ H_{2k+1}}{2k} 
\underset{k \rightarrow \infty}{\longrightarrow}
4c^2. \label{eq_LLN}
\end{align}
 We get 
$$\frac{V_{H_{2k+1}}}{ \sqrt{ H_{2k+1}}}
= \frac{V_{H_{2k+1}}}{c\sqrt{8k}} 
 \cdot  2c \left( \frac{ H_{2k+1}}{2k}   \right)^{-1/2} $$
where the first factor converges to a standard normal distribution by Proposition \ref{prop_CLT} and the second factor converges almost surely to $1$ by \eqref{eq_LLN}. By Slutsky's thorem, we get the following: 
\begin{proposition}
When $k$ tends to infinity, 
$$\frac{V_{H_{2k+1}}}{\sqrt{H_{2k+1}}} \longrightarrow \mathcal{N}(0,1)$$
in distribution. 
\end{proposition}
This result can be compared to the central limit theorem
$$\frac{V_{t}}{\sqrt{t}} \underset{t \rightarrow \infty}{\longrightarrow} \mathcal{N}(0,1),$$
 which is a direct consequence of the inequality 
$$\beta_t - c \leq V_t \leq \beta_t + c,$$
where $(\beta_t)_{t \geq 0}$ is a Brownian motion.

\bibliographystyle{abbrv}
\bibliography{Reference}

\end{document}